\documentclass[letterpaper, 10 pt, journal, twoside]{IEEEtran}
\usepackage{cite}  
\usepackage{amsmath,amssymb,amsfonts}
\usepackage{comment}
\usepackage{xcolor}
\usepackage{multicol}
\usepackage{svg}
\usepackage{mathtools,mathbbol}
\usepackage{multicol,capt-of}
\usepackage{hhline}
\usepackage{multirow}
\usepackage[ruled]{algorithm2e}
\usepackage{algpseudocode}
\usepackage{graphicx}
\usepackage{textcomp}
\usepackage{xcolor}
\usepackage{colortbl}
\usepackage{ragged2e,adjustbox}
\usepackage{bm}
\usepackage{xr,framed}
\usepackage{nicefrac}
\allowdisplaybreaks[3]

\newtheorem{theorem}{Theorem}
\newtheorem{remark}{Remark}

\newtheorem{proof}{Proof}

\newtheorem{lemma}{Lemma}

\newcommand{\final}[1]{{\color{black} #1}}
\newcommand\mydots{\ifmmode\ldots\else\makebox[0.1em][c]{.\hfil.\hfil.}\thinspace\fi}
\allowdisplaybreaks

\def\BibTeX{{\rm B\kern-.05em{\sc i\kern-.025em b}\kern-.08em
    T\kern-.1667em\lower.7ex\hbox{E}\kern-.125emX}}

\RestyleAlgo{ruled}

\begin{document}
\title{
Inverse Optimal Control with Constraint Relaxation
}

\author{Rahel Rickenbach, Amon Lahr, and Melanie N. Zeilinger
\thanks{This work has been supported by the European Union’s Horizon 2020 research and innovation programme, Marie Sklodowska-Curie grant agreement No. 953348, ELO-X.}
\thanks{The datasets generated and/or analysed are available in the eth research collection repository, 10.3929/ethz-b-000745683}
\thanks{All authors are with the Institute of Dynamic Systems and Control, ETH Zurich, Zurich, Switzerland. Email correspondence to: {\tt\small \{rrahel,amlahr,mzeilinger\}@ethz.ch}}}

\maketitle
\thispagestyle{empty}

\begin{abstract}
Inverse optimal control (IOC) is a promising paradigm for learning and mimicking optimal control strategies from capable demonstrators, or gaining a deeper understanding of their intentions, by estimating an unknown objective function from one or more corresponding optimal control sequences. When computing estimates from demonstrations in environments with safety-preserving inequality constraints, acknowledging their presence in the chosen IOC method is crucial given their strong influence on the final control strategy. However, solution strategies capable of considering inequality constraints, such as the inverse Karush-Kuhn-Tucker approach, rely on their correct activation and fulfillment; a restrictive assumption when dealing with noisy demonstrations. To overcome this problem, we leverage the concept of exact penalty functions for IOC and show preservation of estimation accuracy. Considering noisy demonstrations, we then illustrate how the usage of penalty functions reduces the number of unknown variables and how their approximations enhance the estimation method’s capacity to account for wrong constraint activations within a polytopic-constrained environment. The proposed method is evaluated for three systems in simulation, outperforming traditional relaxation approaches for noisy demonstrations.
\end{abstract}

\begin{IEEEkeywords}
Constrained Control, Optimal Control, Uncertain Systems 
\end{IEEEkeywords}

\section{Introduction}
\subsection{Motivation}
\IEEEPARstart{O}{ptimization-based} control strategies, such as model predictive control (MPC) \cite{kouvaritakis2016}, have proven to be effective in performing complex tasks in constrained environments. However, while their performance often heavily depends on a suitably designed objective function, translating a complex goal description into an objective function can be challenging and its tuning delicate. Inverse optimal control (IOC) methods \cite{abazar2020} address this issue by estimating missing parameters of a partially known cost function from optimal control sequences. However, when relying on control sequences obtained in an environment with constraints, the constraints' impact on the final control strategy can be significant. Accordingly, neglecting their influence in the chosen IOC method can result in high estimation errors of the unknown parameters.  
Nonetheless, while equality constraints are widely accepted in most IOC methods, inequality constraints are only considered in certain cases, e.g., in~\cite{englert2017,menner2019}, as their presence results in a more complex problem formulation. Furthermore, building upon the Karush-Kuhn-Tucker (KKT) conditions, these approaches rely on correct constraint activation and fulfillment, which describes a limiting assumption, considering that numerous available control sequences may be suboptimal in practice, for instance, due to their human-generated nature or the presence of measurement noise, and hence might not share the same active set. 
In this paper we \final{relax this limiting assumption} by leveraging the concept of exact penalty functions within the inverse KKT objective estimation method. 
\subsection{Related Work}
Estimation of unknown objective functions from optimal demonstrations, considered in IOC and inverse reinforcement learning (IRL) has been widely addressed in the literature, see e.g. the review in \cite{abazar2020}. While the terms IRL and IOC are often used interchangeably, their problem description could serve as a potential distinction. Apart from some deterministic considerations, e.g., in \cite{wu2023human}, IRL often models the data generated by a Markov Decision Process, which directly allows for the consideration of sub-optimal demonstrations. Well-known solution approaches build upon entropy maximization \cite{ziebart2008,levine2012continuous}, or Bayesian principles \cite{lopes2009active}. While further investigations focus, e.g., on the inverse problem's convexity \cite{garrabe2023} or model-free IRL \cite{sun2025inverse, lian2023inverse}, the consideration of safety-preserving inequality constraints has only recently gained interest \cite{fischer2021sampling,tschiatschek2019learner}. Stochasticity of the dynamics model in IOC is investigated in \cite{nakano2023inverse,mohajerin2018data,garrabe2023probabilistic,garrabe2025convex,yu2021system}. While state constraints are not explicitly considered in \cite{nakano2023inverse,yu2021system,garrabe2023probabilistic}, they are enforced via penalization in \cite{garrabe2025convex}, as well as via a feasible set in \cite{mohajerin2018data}. However, \cite{mohajerin2018data} is limited to convex optimal control problems. Many results in IOC have focused on a deterministic perspective, where, e.g., usage of the KKT optimality conditions, as done in \cite{keshavarz2011imputing},\cite{englert2017} and \cite{menner2019}, allows for a natural consideration of inequality constraints. To account for noisy demonstrations, these works resort to a relaxation of the stationarity condition. However, the influence of potential changes in the set of active constraints, due to measurement noise, is commonly not addressed. A combined estimation of optimal demonstrations and unknown objective parameters, as in \cite{menner2020} and \cite{rickenbach2023inverse}, reduces estimation errors in this case, but relies on the solutions of a mathematical problem with complementarity constraints, which are challenging for most available \mbox{solvers \cite{kanzow2013}}.

\subsection{Contribution}
The contributions of this paper are twofold. By exploiting the concept of exact penalty functions \cite{fletcher2000,kerrigan2000}, we investigate its usage in IOC and show that exact penalties can be leveraged to formulate the original problem as an unconstrained optimal control problem (OCP) with a modified cost while preserving the corresponding IOC estimate. Moreover, considering noisy demonstrations, we illustrate how this reformulation permits a reduction in the number of estimation variables while alleviating the strict selection mechanism of active and inactive inequality constraints in the traditional inverse KKT method via relaxation of the primal feasibility and complementary slackness conditions. This enhances the proposed estimation method's capacity to account for noisy demonstrations within a constrained environment. 
Introducing the problem statement and preliminaries in Section~\ref{sec:problemstatement}, the proposed method is presented in Section~\ref{sec:iocwithconstraintrelaxation} and extended to noisy demonstrations in Section~\ref{sec:iocwithexactpenaltyfunctionsandnoisydemonstrations}. Finally, we conclude this paper with experimental results on three different simulated systems in Section~\ref{sec:experimentalresults}.
\section{Notation}
\label{sec:notation}
The set of real numbers is indicated with $\mathbb{R}$ and $\mathcal{N}(\mu,\Sigma)$ refers to a normal distribution with mean $\mu \in \mathbb{R}^{c}$ and covariance $\Sigma \in \mathbb{R}^{c \times c}$. Given a vector $a \in \mathbb{R}^{c}$ and a matrix $Q \in \mathbb{R}^{c \times c}$ we refer with $a[i]$ to the $i^{th}$, and with $a[0:i]$ to the first $i$ elements. With $Q[i,j]$ we refer to the element located in row $i$ and column $j$, and with \mbox{$Q[i,:]$} and $Q[:,j]$, we reference the $i^{th}$ row and the $j^{th}$ column, respectively. Stacking vectors horizontally is indicated using round brackets, i.e., $b = (a^{\top}, a^{\top})^{\top} \in \mathbb{R}^{2c}$, and $\{a\}$ returns the vector's size, i.e., its number of elements. Defining a function \mbox{$h(a): \mathbb{R}^{b} \rightarrow \mathbb{R}$}, then $\frac{\partial h(a)}{\partial a[j]}$ defines its partial derivative with respect to $a[j]$. Finally, the matrix $\mathbb{O}_{n,m}$  describes a zero matrix of dimensions $n \times m$, and $\odot$ indicates a Hadamard product.   

\section{Problem Statement}
\label{sec:problemstatement}
In this section we introduce the considered OCP and present the inverse KKT approach, which describes a suitable IOC method for optimal demonstrations in constrained environments. We then provide details on the structure of the available noisy demonstrations and comment on an established relaxation method of the inverse KKT estimator.
\subsection{Forward Problem}
\label{subsec:forwardproblem}
\noindent We consider a known, deterministic, discrete-time system
\begin{equation}
    x_{k+1} = f(x_k,u_k),
\end{equation}
where $x_k \in \mathbb{R}^n$ indicates the state and $u_k \in \mathbb{R}^m$ the input vector at time step $k$, generating state and input trajectories, $X = (x_0^{\top},..., x_{N}^{\top})^{\top} \in \mathbb{R}^{n(N+1)}$ and \mbox{$U = (u_0^{\top},..., u_{N-1}^{\top})^{\top} \in \mathbb{R}^{mN}$}, with length $N+1$ and $N$, respectively. Furthermore, to later on simplify the notation, we stack input and state vectors for the first $N$ time steps, defining a trajectory vector \mbox{$T = (t_0^{\top},\dots,t_{N-1}^{\top})^{\top} \in \mathbb{R}^{(n+m)N}$}, where \mbox{$t_k = (x_{k}^{\top},u_{k}^{\top})^{\top} \in \mathbb{R}^{n+m}$} for $k=0,\dots,N-1$. 
The considered OCP, in the following referred to as the \textit{forward problem}, follows as 
\begin{subequations}
\begin{align}
    \min_{U, X} \  &\sum_{k = 0}^{N-1} \theta^{\top} \phi(x_k,u_k) \label{eq:ocpcost}\\
    \mathrm{s.t.} \ & g(x_k,u_k)\leq \mathbb{O}_{p,1} \label{eq:ocp1}\\
    &x_{k+1} = f(x_k,u_k) \label{eq:ocp2} \\
    & x_0 = x(0) \label{eq:ocp3}\\
    & k = 0,...,N-1. \label{eq:ocp4}
\end{align}
\label{eq:optimizationproblem}
\end{subequations}
It is of horizon length $N$ and contains inequality constraints \mbox{$g(x_k,u_k) \leq \mathbb{O}_{p,1}$}, of dimension $p$ in \eqref{eq:ocp1}. Furthermore, its objective function in \eqref{eq:ocpcost} consists of an unknown coefficient vector \mbox{$\theta \in \mathbb{R}^q$} and a known feature vector \mbox{$\phi(x,u) \in \mathbb{R}^q$}. Note that the considered forward-problem design is fairly standard and follows, e.g., the line of \cite{akhtar2021learning,menner2019,garrabe2023probabilistic}.
The optimal input, state and trajectory vector in accordance to OCP~\eqref{eq:optimizationproblem} are indicated by \mbox{$X^* \in \mathbb{R}^{n(N+1)}$}, \mbox{$U^* \in \mathbb{R}^{mN}$}, and $T^*  \in \mathbb{R}^{(n+m)N}$, with \mbox{$X^* = (x_0^{*\top}, ..., x_{N}^{*\top})^{\top}$}, \mbox{$U^* = (u_0^{*\top}, ..., u_{N-1}^{*\top})^{\top}$}, and \mbox{$T^* = (t_0^{*\top},\dots,t_{N-1}^{*\top})^{\top}$}, respectively, and it is assumed that $f$, $g$, and $\phi$ are twice continuously differentiable.
Terminal ingredients, like terminal costs and constraints, are omitted for simplicity, however, all subsequent derivations allow for their inclusion.

\subsection{The Inverse Karush-Kuhn-Tucker Approach}
\label{subsec:inverseKKTapproach}
The goal of IOC is to provide an estimate of $\theta$, solely relying on available optimal demonstrations $X^*$ and $U^*$. In a constrained state and input space, such an estimate can, for example, be obtained using the inverse KKT approach, which is presented in \cite{englert2017,menner2019} and detailed in the following.

\noindent Introducing the Lagrangian multipliers \mbox{$v_k \in \mathbb{R}^n$}, \mbox{$\lambda_{k} \in \mathbb{R}^p$} and summarizing them for \mbox{$k = 0,...,N-1$} in the respective vectors \mbox{$\lambda \in  \mathbb{R}^{pN}$}, \mbox{$v \in  \mathbb{R}^{nN}$}, with $\lambda = (\lambda_{0}^{\top}, ..., \lambda_{N-1}^{\top})^{\top}$ and \mbox{$v = (v_{0}^{\top}, ..., v_{N-1}^{\top})^{\top}$}, the Lagrangian of \eqref{eq:optimizationproblem} follows as
\begin{equation}
\begin{split}
    & \mathcal{L}(\theta,\lambda,v,X,U) = \\
    & \sum_{k = 0}^{N-1} \theta^{\top} \phi(x_k,u_k) +
    \begin{pmatrix}
    v_k \\ 
    \lambda_k
    \end{pmatrix}^{\top} 
    \begin{pmatrix}
    x_{k+1} - f(x_k,u_k) \\ 
    g(x_k,u_k) 
    \end{pmatrix}.
\end{split}
\end{equation}
As the KKT conditions constitute necessary optimality requirements for any solution of an optimization problem that fulfills suitable constraint qualifications \cite{kuhn2014}, their induced system of equations can be solved for $\theta$, $v$, and $\lambda$. Hence, defining $b=n(N+1)+mN$, then for all $k = 0,...,N-1$ the conditions read as
\begin{subequations}
\begin{align}
    &\nabla_{X,U}\mathcal{L}(\theta,\lambda,v,X,U)\vert_{X = X^{*}, U = U^{*}} = \mathbb{O}_{b,1} \label{eq:kkt1}\\
    & \lambda_{k} \odot g(x_k^{*},u_k^{*}) = \mathbb{O}_{p,1} \label{eq:kkt2}\\ 
    & g(x_k^*,u_k^*)  \leq \mathbb{O}_{p,1} , x_{k+1}^{*} = f(x_k^*,u_k^*) \label{eq:kkt3b} \\
    &\lambda_{k} \geq \mathbb{O}_{p,1}. \label{eq:kkt4}
\end{align}
\label{eq:kktsyseq}
\end{subequations}
Note that the primal feasibility conditions~\eqref{eq:kkt3b} do not provide information on any estimation variable of interest, are trivially fulfilled for the optimal demonstrations and hence neglected in the inverse KKT approach.\\
For simplification and comparisons with reformulations in upcoming sections, we leverage the considered problem description, which is linear in the unknown objective parameter $\theta$, 
and transform the stationarity condition in equation~\eqref{eq:kkt1} into the following regression problem
\begin{equation}
    \begin{pmatrix}
    J_{\theta}^{\top}  & J_{v}^{\top} & J_{\lambda}^{\top}
    \end{pmatrix}\vert_{X=X^*,U=U^*}
    \begin{pmatrix}
    \theta  \\ 
    v \\
    \lambda
    \end{pmatrix}
    = \mathbb{O}_{b,1},
    \label{eq:linearsystemofequation}
\end{equation}
where the Jacobian of the summed-up objective feature vector, equality and inequality constraints are denoted as $ J_{\theta} = \frac{\partial \sum_{k = 0}^{N-1}\phi(x_k,u_k)}{\partial(X^{\top},U^{\top})}$, $J_{v} = \frac{\partial F(X,U)}{\partial(X^{\top},U^{\top})}$, and $J_{\lambda} = \frac{\partial G(X,U)}{\partial(X^{\top},U^{\top})}$ respectively, with 
$
    F(X,U) = ((x_{1} - f(x_0,u_0))^{\top},...,(x_{N} - f(x_{N-1},u_{N-1}))^{\top})^{\top}
$
and 
$
G(X,U) = (g(x_0,u_0)^{\top},...,g(x_{N-1},u_{N-1})^{\top})^{\top}.
$
Since the regression problem describes a necessary optimality condition, it has at least one solution. Underdetermination can often be addressed by leveraging the complementary slackness constraint in equation \eqref{eq:kkt2}, to eliminate the multipliers associated with inactive inequality constraints and accordingly achieve full column rank. 
Consequently, it divides the vector $\lambda$ into a vector $\tilde{\lambda}$ containing all timesteps and respective dimensions at which \mbox{$g(x_k^*,u_k^*) \leq \mathbb{O}_{p,1}$} is active, i.e., employing some suitable reordering of $g$ such that \mbox{$g(x_k^*,u_k^*)[0:\{ \tilde{\lambda}\}] = \mathbb{O}_{\{ \tilde{\lambda}\},1}$}, as well as a vector $\check{\lambda}$ containing all timesteps at which \final{\mbox{$g(x_k^*,u_k^*) \leq \mathbb{O}_{p,1}$}} is not active, i.e., \mbox{$g(x_k^*,u_k^*)[\{ \tilde{\lambda}\}:p] < \mathbb{O}_{p-\{ \check{\lambda}\},1}$}, and sets the multipliers~$\tilde{\lambda}$ of the latter equal to zero. Applying the same separation to the rows of $J_{\lambda}$, we obtain $J_{\tilde{\lambda}}$ and $J_{\check{\lambda}}$, respectively, and the problem in \eqref{eq:kktsyseq} reduces to:
\begin{align}
    & 
    \begin{pmatrix}
    J_{\theta}^{\top}  & J_{v}^{\top} & J_{\tilde{\lambda}}^{\top} & J_{\check{\lambda}}^{\top}
    \end{pmatrix}\vert_{X=X^*,U=U^*}
    \begin{pmatrix}
    \theta  \\ 
    v \\
    \tilde{\lambda} \\
    \mathbb{O}_{\{\check{\lambda}\},1}
    \end{pmatrix}
    = \mathbb{O}_{b,1} \nonumber \\
    & \ \ \mathrm{s.t.} \ \ \tilde{\lambda}_k \geq \mathbb{O}_{p,1} \ \ \forall k = 0,...,N-1.
    \label{eq:linearsystemofequationactiveconstraints}
\end{align}

\subsection{Noisy Demonstrations}
\label{subsec:availabledemonstrations}
Acknowledging that truly optimal demonstrations are often unavailable, we consider a set of $D$ demonstrations, corrupted with measurement noise on state and input. The measurement noise is assumed to be independent between time steps, but potentially cross-correlated between states and inputs. \final{Therefore, we leverage the trajectory notation~$T$, and indicate the joint measurement noise, acting on the stacked state and input vectors,} with $w^t_k$. It is distributed as a zero-mean Gaussian with standard deviation \mbox{$\Sigma^t_k\in \mathbb{R}^{(n+m) \times (n+m)}$}, i.e., \mbox{$w^t_k \sim \mathcal{N}(\mathbb{O}_{(n+m),1}, \Sigma^t_k)$} for \mbox{$k = 0,...,N-1$}. Indicating its realization in the $d$-th demonstration with $w^t_{k,d}$ the corresponding measurements follow as $t_{k,d} = t_{k}^* + w^t_{k,d}$
and are summarized for each demonstration in \mbox{$T_d \in \mathbb{R}^{(n+m)N}$} with \mbox{$T_d = (t_{0,d}^{\top},\dots,t_{N-1,d}^{\top})^{\top}$}. Separating into states and input sequences, we correspondingly obtain \mbox{$U_d  \in \mathbb{R}^{mN}$}, with \mbox{$U_d = (u_{0,d}^{\top}, ..., u_{N-1,d}^{\top})^{\top}$}, and \mbox{$X_d \in \mathbb{R}^{n(N+1)}$}, with \mbox{$X_d = (x_{0,d}^{\top}, ..., x_{N,d}^{\top})^{\top} $}, where we additionally define $x_{N,d}^{\top} = x_{N}^* + w^x_{N,d}$. Here we denote $w^x_{N,d}$ as the $d$-th realization of the measurement noise $w^x_N$ on $x_N^*$, which is distributed as a zero-mean Gaussian with standard deviation \mbox{$\Sigma^x_N\in \mathbb{R}^{n \times n}$}, i.e.,  \mbox{$w^x_N \sim \mathcal{N}(\mathbb{O}_{n,1}, \Sigma^x_N)$}.

\subsection{Relaxation of the inverse KKT Approach}
\label{subsec:relaxationofinversekkt}
Working with noisy demonstrations, i.e., considering the set of $X_d$ and $U_d$, with $d=1,...,D$, the KKT equations presented in \eqref{eq:kkt1}-\eqref{eq:kkt4} might be violated and the resulting system of equations potentially has no solution. In the inverse KKT approach \cite{englert2017,menner2019}, \cite{keshavarz2011imputing} this issue is accounted for with a least-squares relaxation of the stationarity condition in \eqref{eq:kkt1}, resulting in the following optimization problem:
\begin{subequations}
\begin{align}
    \min_{\substack{\theta,v_{1:D} \\ \lambda_{1:D}}} \ &\sum_{d=1}^{D}\Vert \nabla_{X,U} \mathcal{L}(\theta,\lambda_d,v_d,X,U)\vert_{X = X_d, U = U_d}\Vert_2^2 \label{eq:kkt1ls}\\
    \mathrm{s.t.} \ &\lambda_{k,d} \odot g(x_{k,d},u_{k,d}) = \mathbb{O}_{p,1} \label{eq:kkt2ls}\\ 
    & \lambda_{k,d} \geq \mathbb{O}_{p,1} \label{eq:kkt4ls} \\
    & k=0,...,N-1, \ \ d=1,...,D. \label{eq:kkt5ls}
\end{align}
\label{eq:kktls}
\end{subequations}
While this approximation considers violations of the stationarity condition, note that, considering noisy demonstrations as introduced in Section~\ref{subsec:availabledemonstrations}, constraint violations are possible. However, potential violations of primal feasibility in equation \eqref{eq:kkt2}, as well as wrongful constraint inactivation are commonly not accounted for. 
Furthermore, the probability of exactly activating an inequality constraint, i.e., obtaining a noise combination $w^t_{k,d}$ such that \mbox{$g(x_{k,d},u_{k,d}) = \mathbb{O}_{p,1} $}, is zero. Hence, with probability one, usage of the complementary slackness relation, as shown in problem description~\eqref{eq:linearsystemofequationactiveconstraints} of Section~\ref{subsec:inverseKKTapproach}, results in \mbox{$\lambda_{k,d} = \mathbb{O}_{p,1}$} for all $d = 1, ..., D$ and $k = 0, ..., N-1$, representing an unconstrained IOC problem, which, given the importance of inequality constraints for an optimization problem's solution, can result in large estimation errors. An intuitive approach to address this issue under the assumption that the optimal solution satisfies the constraints, is to replace \eqref{eq:kkt2ls} with 
\begin{align}
    \lambda_{k,d} \odot \min(\mathbb{O}_{p,1},g(x_{k,d},u_{k,d})) = \mathbb{O}_{p,1}, \label{eq:kkt2trunc}
\end{align}
\label{eq:kkttrunc}
and truncate the constraint value with zero.
This can result in an increased estimation performance; however, it neither offers a strong theoretical foundation nor provides a solution for noise realizations \emph{inactivating} an inequality constraint that would have been active under the optimal solution. 
\begin{remark}
    Note that the inverse KKT approach with noisy data describes an errors-in-variables problem, as investigated in \cite{rickenbach2023inverse}. Accordingly, a least-squares estimate is generally biased. Nevertheless, it has shown good empirical approximation behavior for sufficiently small noise realizations (see e.g., \cite{menner2019}, \cite{keshavarz2011imputing}) and is therefore also considered in this paper.
\end{remark}
\section{IOC with Exact Penalty Functions and Optimal Demonstrations}
\label{sec:iocwithconstraintrelaxation}
In the following, we propose a reformulation of the IOC problem in \eqref{eq:kktsyseq} that leverages the concept of exact penalty functions, preserves estimation accuracy under the availability of optimal demonstrations, and serves as a basis for the subsequent constraint relaxation when considering noisy demonstrations.
For this purpose, we introduce exact penalty functions for OCPs in the first subsection and address their employment for IOC in the following one. 

\subsection{Exact Penalty Functions in Optimal Control Problems}
The method of exact penalty functions, e.g., presented in \cite{fletcher2000,kerrigan2000}, allows for the inclusion of inequality constraints in the objective of an optimal control problem. This is achieved by replacing the respective inequality constraint, represented by $g(x,u) \leq \mathbb{O}_{p,1}$,  with an L1-penalization of constraint violation, i.e., $g_{\max}(x,u) \in \mathbb{R}^{p}$ where
$
    g_{\max}(x,u) = (\vert\max(0,g(x,u)[1])\vert,...,\vert\max(0,g(x,u)[p])\vert)^{\top}.
$ 
This penalization is further scaled with $\rho_k \in \mathbb{R}^{p}$ , resulting in a reformulation of OCP \eqref{eq:optimizationproblem} that reads as
\begin{subequations}
\begin{align}
    \min_{U, X} \ &\sum_{k = 0}^{N-1} \theta^{\top} \phi(x_k,u_k) + \rho_k^{\top} g_{\max}(x_k,u_k) \\
    \mathrm{s.t.} \ & \eqref{eq:ocp2}-\eqref{eq:ocp4}.
\end{align}
\label{eq:optimizationproblemexactpenalty}
\end{subequations}
\noindent Denoting with $\lambda_k^*$ the solution of the corresponding dual problem of \eqref{eq:optimizationproblem}, which is equal to the solution of \eqref{eq:kktsyseq}, we can follow the line of \cite{fletcher2000,kerrigan2000} and restate Lemma~\ref{therorem:exactpenalty}:

\begin{lemma}[cf. \cite{kerrigan2000}, Thm. 1]
    If the penalty weight $\rho_k[j]$ is bigger or equal to $\Vert \lambda_k^*[j] \Vert_{\infty}$, for all $j=1,...,p$, then any solution of \eqref{eq:optimizationproblem} is a solution of~\eqref{eq:optimizationproblemexactpenalty}.
    \label{therorem:exactpenalty}
\end{lemma} 

\subsection{Exact Penalty Functions in Inverse Optimal Control}

Not solving the OCP in \eqref{eq:optimizationproblem}, but working with measurements of its solution, the dual solution $\lambda_k^*$, and hence the values of $\rho_k$ for which $X^*$ and $U^*$ describe a solution of~\eqref{eq:optimizationproblemexactpenalty}, are unknown for all $k = 0, ..., N-1$. Thus, to incorporate the idea of exact penalty functions into the IOC problem, we extend the vector of unknown objective parameters by $\rho = (\rho_0, ..., \rho_{N-1})^{\top}$ and the considered OCP follows as 
\begin{equation}
\begin{split}
    \min_{U, X} \ & 
        \sum_{k = 0}^{N-1} \begin{pmatrix} \theta \\ \rho_k \end{pmatrix}^{\top} \begin{pmatrix} \phi(x_k,u_k) \\ g_{\max}(x_{k},u_{k}))  \end{pmatrix}  \\
    \mathrm{s.t.} \ & \eqref{eq:ocp2}-\eqref{eq:ocp4}.
\end{split}
\label{eq:optimizationproblemexactpenaltyioc}
\end{equation}

\noindent IOC now aims at estimating the extended parameter vector $\bar{\theta} = (\theta, \rho) \in \mathbb{R}^{q+Np}$. 
Building upon the KKT conditions of OCP \eqref{eq:optimizationproblemexactpenaltyioc}, and constructing the Lagrangian as
\begin{align}
    & \mathcal{L}_{ep}(\theta,\rho,v,X,U) = \\
    & \sum_{k = 0}^{N-1} \begin{pmatrix}
    \theta \\ 
    \rho_k
    \end{pmatrix}^{\top} \begin{pmatrix} \phi(x_k,u_k) \\ g_{\max}(x_{k},u_{k}) \end{pmatrix} + v_k^{\top} \begin{pmatrix}
    x_{k+1} - f(x_k,u_k) 
    \end{pmatrix} \nonumber,
\end{align}
such an estimate can be obtained as the solution of the following system of equations:
\begin{subequations}
\begin{align}
    &\nabla_{X,U}\mathcal{L}_{ep}(\theta,\rho,v,X,U)\vert_{X = X^{*}, U = U^{*}} =  \mathbb{O}_{b,1}, \label{eq:kktep1}\\
    &  \rho_{k} \geq \mathbb{O}_{p,1} \ \ \forall k = 0,...,N-1. \label{eq:kktep2}
\end{align}
\label{eq:kktepsyseq}
\end{subequations}
Note that, as a consequence of Lemma \ref{therorem:exactpenalty}, in \eqref{eq:kktep2} we additionally included the knowledge of each element of $\rho_k$ being positive. Like in the traditional inverse KKT approach, the problem description remains linear in the unknown objective parameter $\bar{\theta}$. Defining the derivative of the $\max$ operator as its right-hand derivative, we indicate with \mbox{$J_{\rho}  = \frac{\partial G_{\max}(X,U)}{\partial(X^{\top},U^{\top})}$} the Jacobian of the feature vector extension, where $G_{\max}(X,U) = (g_{\max}(x_0,u_0)^{\top},...,g_{\max}(x_{N-1},u_{N-1})^{\top})^{\top}$, it can be transformed into the following constrained regression problem:
\begin{align}
        & \underbrace{\begin{pmatrix}
        J_{\theta}^{\top}  & J_{v}^{\top} & J_{\rho}^{\top}
        \end{pmatrix}}_{J_{ep}(X,U)} \vert_{X=X^*,U=U^*}
        \underbrace{\begin{pmatrix}
        \theta  \\ 
        v \\
        \rho
        \end{pmatrix}}_{\beta_{ep}} \nonumber
        = \mathbb{O}_{b,1} \\ & \ \ \mathrm{s.t.} \ \rho_k \geq \mathbb{O}_{p,1} \ \ \forall k=0,...,N-1. 
        \label{eq:linearsystemofequationexactpenaltynotseparated}
\end{align} 
Considering the obtained estimate, we formulate Theorem~\ref{theorem:exactnessofsolution}, stating that the estimate of the unknown objective parameter~$\theta$ from optimal demonstrations $X^*$ and $U^*$ of the inverse-KKT formulation with penalty functions \eqref{eq:linearsystemofequationexactpenaltynotseparated} is equal to the estimate obtained from the original inverse-KKT formulation~\eqref{eq:linearsystemofequationactiveconstraints}.

\begin{theorem} 
    Let $\theta_1,v_1$, and $\lambda_1$ be the solution of the estimation problem in \eqref{eq:linearsystemofequationactiveconstraints} and $\theta_2,v_2$, and $\rho_2$ the solution of the estimation problem in \eqref{eq:linearsystemofequationexactpenaltynotseparated}. Then it holds that $\theta_1 = \theta_2$.
    \label{theorem:exactnessofsolution} 
\end{theorem}
\begin{proof}
Theorem \ref{theorem:exactnessofsolution} is proven, by comparing the inverse KKT formulation of a problem formulation with and without penalization, both relying on a clear separation of active and inactive inequality constraints.
Hence, following the line of the introduced separation of the  $\lambda$-vector in Section~\ref{subsec:inverseKKTapproach}, we define an identical separation of $\rho$ into a vector $\tilde{\rho}$, containing all timesteps and respective dimensions at which \mbox{$g(x_k,u_k) \leq \mathbb{O}_{p,1}$} is active, i.e., employing some suitable reordering of $g$ such that \mbox{$g(x_k^*,u_k^*)[0:\{ \tilde{\rho}\}] = \mathbb{O}_{\{ \tilde{\rho}\},1}$}, as well as a vector $\check{\rho}$, containing all timesteps at which \mbox{$g(x_k,u_k) \leq \mathbb{O}_{p,1}$} is not active, i.e., \mbox{$g(x_k^*,u_k^*)[\{ \tilde{\rho}\}:p] < \mathbb{O}_{\{ \check{\rho}\},1}$}. Additionally, we denote with $J_{\tilde{\rho}}$ and $J_{\check{\rho}}$ the corresponding separation of  $J_{\rho}$.
Then, considering the estimation problem in \eqref{eq:linearsystemofequationexactpenaltynotseparated}, it can be seen that the separation into $J_{\tilde{\rho}}$ and $J_{\check{\rho}}$ happens implicitly under the usage of the $\max(\cdot)$ operator, setting all elements of $J_{\check{\rho}}$ to zero. 
Including the positivity constraint in \eqref{eq:kktep2}, the estimation problem follows as:
\begin{align}
    & \begin{pmatrix}
    J_{\theta}^{\top}  & J_{v}^{\top} & J_{\tilde{\rho}}^{\top} & \mathbb{O}_{b,\{\check{\rho}\}}
    \end{pmatrix} \vert_{X=X^*,U=U^*}
    \begin{pmatrix}
    \theta  \\ 
    v \\
    \tilde{\rho} \\
    \check{\rho}
    \end{pmatrix} \nonumber
    = \mathbb{O}_{b,1} \\ &  \mathrm{s.t.} \ \rho_k \geq 0 \ \ \forall k=0,...,N-1. 
    \label{eq:linearsystemofequationexactpenalty}
\end{align} 
From its definition, it follows that \mbox{$J_{\tilde{\rho}}$} is equal to  $J_{\tilde{\lambda}}$. Furthermore, it can be seen that $\check{\rho}$ does not enter the regression problem, as each of its values is multiplied by zero. Therefore, all solutions of \eqref{eq:linearsystemofequationactiveconstraints} and \eqref{eq:linearsystemofequationexactpenalty} are identical in $\theta$, $v$, and $ \tilde{\rho} = \tilde{\lambda}$, but may differ in $\check{\rho} \neq \check{\lambda}$. \hspace{5.1em}\hfill\ensuremath{\blacksquare}
\end{proof}
\section{IOC with Exact Penalty Functions and Noisy Demonstrations}
\label{sec:iocwithexactpenaltyfunctionsandnoisydemonstrations}
In this section, we extend the proposed reformulation to noisy demonstrations. For the case of polytopic inequality constraints, it will be shown how noisy demonstrations lead to a natural usage of the softmax approximation and correspondingly allow for a constraint relaxation accounting for potential violations of primal feasibility in IOC problems. 
\subsection{Exact Penalization in Inverse KKT Approximation}
Rewriting the exact-penalty reformulation of the inverse KKT estimator, introduced in \eqref{eq:kktep2}, with noisy demonstrations and leveraging the least-squares relaxation of the stationarity condition, similar to \eqref{eq:kktls}, the estimation problem follows as
\begin{subequations}
\begin{align}
    \min_{\beta_{ep}} \ &  
    \Big\|
    \sum_{d=1}^{D} J_{ep}(X_d,U_d) \beta_{ep} 
    \Big\|_2^2
    \\ 
    \mathrm{s.t.} \ & \ \rho_k \geq \mathbb{O}_{p,1}, \ k = 0,...,N-1.
\end{align}
\label{eq:iocep}
\end{subequations}
Note that considering inequality constraints in the OCP's objective not only allows for their inclusion within the least-squares approximation but also reduces unknown estimation variables. Rather than estimating an individual dual parameter vector $\lambda_d$ for each demonstration, we are interested in the parameter vector $\beta_{ep}$, and accordingly the penalty parameter vector $\rho$, which results in the smallest error over all of them.
Indicating a Jacobian's empirical mean with \mbox{$\bar{J}_{(\cdot)} = \frac{1}{D}\sum_{d=1}^{D} J_{(\cdot)}^{\top}(X_d,U_d)$}, the objective of the IOC problem in \eqref{eq:iocep} can be scaled and rewritten as
\vspace{-0.0em}
\begin{align}
    &  \left\Vert \frac{1}{D}\sum_{d=1}^{D} J_{ep}(X_d,U_d) \beta_{ep} \right\Vert_2^2  \label{eq:iocepaveraged}\\  
    & =: \Vert ( \bar{J}_{\theta}^{\top}(X_d,U_d)\theta+\bar{J}_{v}^{\top}(X_d,U_d)v+ \bar{J}_{\rho}^{\top}(X_d,U_d)\rho \Vert_2^2.   \nonumber
\end{align}
\subsection{IOC with Constraint Relaxation}
In order to relax the inequality constraints in the case of noisy demonstrations, we consider polytopic inequality constraints, i.e., \mbox{$g(x,u) = H(x^{\top},u^{\top})^{\top} - h$}, with \mbox{$H \in \mathbb{R}^{p\times(n+m)}$} and \mbox{$h \in \mathbb{R}^{p}$}, and reformulate the mathematical expression of the averaged inequality constraint Jacobian $\bar{J}_{\rho}$. 
We then recognize that the summation of all demonstrations and divisions by their number, as done in \eqref{eq:iocepaveraged}, represents an approximation of the expectation operator and swap differentiation and expectation, utilizing their linearity. Hence, for each dimension $j=1,..., p$ of the polytopic constraint and each timestep $k=0,...,N-1$, the submatrix of $\bar{J_{\rho}}$ containing the corresponding elements follows as 
\vspace{-0.0em}
\begin{align*}
    &\bar{J}_{\rho,jk} = \left. \frac{1}{D}\sum_{d=1}^{D} \left(\frac{\partial}{\partial(t_k^{\top})} \max(0,H[j,:]t_{k}^{\top} - h[j])\right) \right\vert_{t_{k,d}}   \\
    &\approx \left(\frac{\partial}{\partial(t_k^{\top})} (\mathbb{E}_{t_k}[\max(0,H[j,:]t_{k}^{\top} - h[j])]) \right) 
\end{align*}
We then leverage the fact that the measured state and input variables are both normally distributed around their optimal values. 
Hence, following the derivations in \cite{messerer2023}, it can be shown that the expectation of our penalization results in the following approximation of the ReLU function, i.e., $\mathbb{E}_{t_k}[\max(0,H[j,:]t_{k}^{\top} - h[j])] \approx \Psi(\mu_g,\sigma_g)$, where
\begin{equation}
\begin{split}
\Psi(\mu_g,\sigma_g)= \sigma_g p_{\mathcal{N}}\left(\frac{\mu_g}{\sigma_g}\right) + \mu_g P_{\mathcal{N}}\left(\frac{\mu_g}{\sigma_g}\right),\end{split}
\label{eq:expectationrelaxationexactpenalty}
\end{equation}
\noindent with $\mu_g = H[j,:]t_k^*-h[j]$ and $\sigma_g = \sqrt{H[j,:]^{\top}\Sigma^t_kH[j,:]}$. Here, $p_{\mathcal{N}}(\cdot)$ indicates the probability density function (PDF) and $\mathcal{P}_{\mathcal{N}}(\cdot)$ the cumulative distribution function (CDF) of the standard normal distribution. The expectation $t^*_k$ is thereby empirically approximated as $t^*_k \approx \frac{1}{D}\sum_{d=1}^D t_{k,d}$.\\ \final{Therefore, we overcome the stringent selection of a certain constraint being either active or not by replacing the ReLu in $g_{\mathrm{max}}$, which arises from the complementary slackness condition and can lead to high estimation errors when noise realizations lead to wrongful activation or inactivation, with its approximation in equation \eqref{eq:expectationrelaxationexactpenalty}.} Furthermore, note that the amount of relaxation does not need to be heuristically chosen, but follows directly from the amount of noise in the direction of the constraint, i.e., the value of $\sigma_g$. 
\begin{remark} 
    The resulting approximation never results in a function value equal to zero, leading to an underdetermined estimation problem when implemented without a cutoff.  
    To address this issue in our experimental section, we set function values below a certain threshold, in the following experiments chosen as $\Psi(\mu_g,\sigma_g) \leq 0.005$, (and as $\Psi(\mu_g,\sigma_g) \leq 0 $ for the noiseless case, i.e., 0.01\% noise), equal to zero. 
\end{remark}

\section{Experimental Results}
\label{sec:experimentalresults}
The proposed framework is tested on three systems in simulation: a mass-spring-damper system (System 1), an inverted pendulum (System 2), and a kinematic bicycle model (System 3). The continuous system dynamics are discretized using the backward Euler method with a sampling time of \mbox{$0.1$s} and the considered OCPs contain a polytopic inequality constraint $H(x^{\top},u^{\top})^{\top} \leq a$. 
Throughout the experiments, we consider different values for the input noise covariance, which is calculated proportional to the mean input values of the optimal demonstration over the considered horizon, denoted with $u_{m,D}$. 
All estimates obtained according to the traditional least-squares inverse KKT approximation as proposed in \cite{menner2019} and presented in Section~\ref{subsec:relaxationofinversekkt}, equation \eqref{eq:kktls}, are further denoted as $\theta_{\mathrm{KKT}}$, the ones obtained with the truncated heuristics, i.e., replacing \eqref{eq:kkt2ls} with \eqref{eq:kkt2trunc}, as $\theta_{\mathrm{TR}}$, and the ones obtained from our proposed constraint relaxation approach, i.e., combining the estimation problem in \eqref{eq:iocep} with the relaxation in \eqref{eq:expectationrelaxationexactpenalty}, as $\theta_{\mathrm{EP}}$. Their quality is evaluated through the root-mean-square error (RMSE) concerning the true values in $\theta^{*}$. Each experiment is repeated 10 times with $D=10$ demonstrations each, presenting mean and standard deviations of the RMSE in Table \ref{tab:baseanalysis}. In Table \ref{tab:robustnessanalysis} we investigate the method's robustness on the example of the kinematic bicycle model. For this purpose, we assume uncertainty of the car length and, for each experiment, sample its value in the IOC estimation uniformly from an interval of $\pm 5 \%$ around the true value. 

\paragraph{Mass-Spring-Damper System (System 1)}
\noindent The dynamics read as $\dot{x}_1 = x_2$, $\dot{x}_2 = m^{-1}(-cx_1 -dx_2 + u)$,
with $m=1.0$ $\mathrm{kg}$, $c=0.2$ $\mathrm{kg\cdot s^{-2}}$ and $d=0.1$ $\mathrm{kg\cdot s^{-1}}$.
The forward problem has a horizon of $N = 10$, includes the inequality constraint $u \leq \final{0.55}$, and we chose $\theta = (10,5,7)^{\top}$, $\phi(x,u) = ((x_1-3)^{2},(x_2-0)^{2},u^2)^{\top}$, \mbox{$x_0 = (1,0.1)^{\top}$}.
\paragraph{Inverted Pendulum (System 2)}
\noindent The dynamics of an inverted pendulum, with $m=1.0$ kg, $l= 0.8$ m and $g = 9.81$ $\mathrm{Nm^{2}kg^{-2}}$, are $\dot{x}_1 = x_2$, $\dot{x}_2 = -gl^{-1}\sin(x_1) + (ml^{2})^{-1}u.$
The forward problem has a horizon length of $N = 10$ and includes the constraint \mbox{$u \leq 1.90$}. Furthermore, we chose $\theta = (10,5,7)$, $\phi(x,u) = ((x_1-0.5)^{2},(x_2-0.1)^{2},u^{2})^{\top}$, \mbox{$x_0 = (1.5,0.5)^{\top}$}.

\paragraph{Kinematic Bicycle Model (System 3)}
\noindent Representing a more realistic use case, our approach has been tested with the kinematic bicycle model \cite{Rajamani2012}. Defining $L = 0.115$ $\mathrm{m}$, its dynamics read as $\dot{x}_1 = u_1 \cos(x_3)$, $\dot{x}_2 = u_1 \sin(x_3)$, $\dot{x}_3 = {u_1} \tan(x_4)/L$, $\dot{x}_4 = u_2$.
The forward problem has a horizon length of $N = 10$ and includes the inequality constraint \mbox{$u_2 \leq 0.35$}. Furthermore we choose \mbox{$\theta = (10,10,3,3,8,5)$}, $\phi(x,u) = ((x_1-3)^{2},(x_2-3)^{2},(x_3-0)^{2},(x_4-0)^{2},u_1^2,u_2^2)^{\top}$, \mbox{$x_0 = (0,0,0,0)^{\top}$}.
\vspace{-1.0em}
\begin{table}[h]
\caption{RMSE of estimates on Systems 1 up to 3.}
\vspace{-0.8em}
\small
\begin{center}
\begin{adjustbox}{width=0.95\linewidth}
\begin{tabular}{|c|c|c|c|c|c|}
\hline
\multicolumn{2}{|c|}{System}  & \multicolumn{3}{c|}{$\Sigma_{U}$} \\
\cline{3-5}
\multicolumn{2}{|c|}{}                   & 0.01\% $u_{m,D}$ & 5\% $u_{m,D}$  & 10\% $u_{m,D}$ \\
\hline
\rowcolor[HTML]{EFEFEF} 
\cellcolor[HTML]{FFFFFF}
&$\theta_{\mathrm{KKT}}$ & $\mathbf{0.00 \pm 0.00}$ & $4.34 \pm 1.26$ & $4.12 \pm 0.88$ \\
\cline{2-5}
\rowcolor[HTML]{EFEFEF} 
\cellcolor[HTML]{FFFFFF}
&$\theta_{\mathrm{TR}}$ & $\mathbf{0.00 \pm 0.00}$ & $\mathbf{0.39 \pm 0.21}$ & $\mathbf{1.08 \pm 0.63}$ \\
\cline{2-5}
\rowcolor[HTML]{d8eddf}
\multirow{-3}{*}{\cellcolor[HTML]{FFFFFF}1} 
&$\theta_{\mathrm{EP}}$ & $\mathbf{0.00 \pm 0.00}$ & $0.48 \pm 0.26$ & $1.34 \pm 0.77$ \\
\hline
\rowcolor[HTML]{EFEFEF} 
\cellcolor[HTML]{FFFFFF}
&$\theta_{\mathrm{KKT}}$ & $0.01 \pm 0.01$ & $8.60 \pm 0.00$ & $8.60 \pm 0.00$ \\
\cline{2-5}
\rowcolor[HTML]{EFEFEF} 
\cellcolor[HTML]{FFFFFF}
&$\theta_{\mathrm{TR}}$ & $\mathbf{0.00 \pm 0.00}$ & $8.50 \pm 0.23$ & $8.59 \pm 0.02$ \\
\cline{2-5}
\rowcolor[HTML]{d8eddf}
\multirow{-3}{*}{\cellcolor[HTML]{FFFFFF}2} 
&$\theta_{\mathrm{EP}}$ & $\mathbf{0.00 \pm 0.00}$ & $\mathbf{0.86 \pm 0.74}$ & $\mathbf{1.70 \pm 1.28}$ \\
\hline
\rowcolor[HTML]{EFEFEF} 
\cellcolor[HTML]{FFFFFF}
&$\theta_{\mathrm{KKT}}$ & $\mathbf{0.00 \pm 0.00}$ & $372.84 \pm 153.60$ & $434.81 \pm 106.72$ \\
\cline{2-5}
\rowcolor[HTML]{EFEFEF} 
\cellcolor[HTML]{FFFFFF}
&$\theta_{\mathrm{TR}}$ & $\mathbf{0.00 \pm 0.00}$ & $3.19 \pm 2.59$ & $66.70 \pm 173.20$ \\
\cline{2-5}
\rowcolor[HTML]{d8eddf}
\multirow{-3}{*}{\cellcolor[HTML]{FFFFFF}3} 
&$\theta_{\mathrm{EP}}$ & $\mathbf{0.00 \pm 0.00}$ & $\mathbf{0.25 \pm 0.25}$ & $\mathbf{0.64 \pm 0.54}$ \\
\hline
\end{tabular}
\end{adjustbox}
\end{center}
\vspace{-1.5em}
\label{tab:baseanalysis}
\end{table}
\normalsize

\begin{table}[h]
\caption{Robustness analysis on System 3.}
\vspace{-0.7em}
\small
\begin{center}
\begin{adjustbox}{width=0.95\linewidth}
\begin{tabular}{|c|c|c|c|c|c|}
\hline
\multicolumn{2}{|c|}{System}  & \multicolumn{3}{c|}{$\Sigma_{U}$} \\
\cline{3-5}
\multicolumn{2}{|c|}{}                   & 0.01\% $u_{m,D}$ & 5\% $u_{m,D}$  & 10\% $u_{m,D}$ \\
\hline
\rowcolor[HTML]{EFEFEF} 
\cellcolor[HTML]{FFFFFF}
&$\theta_{\mathrm{KKT}}$ & $2.74 \pm 2.54$ & $370.99 \pm 154.95$ & $439.15 \pm 105.85$ \\
\cline{2-5}
\rowcolor[HTML]{EFEFEF} 
\cellcolor[HTML]{FFFFFF}
&$\theta_{\mathrm{TR}}$ & $2.74 \pm 2.54$ & $9.73 \pm 10.12$ & $65.42 \pm 174.24$ \\
\cline{2-5}
\rowcolor[HTML]{d8eddf}
\multirow{-3}{*}{\cellcolor[HTML]{FFFFFF}3} 
&$\theta_{\mathrm{EP}}$ & $\mathbf{0.45 \pm 0.27}$ & $\mathbf{0.61 \pm 0.37}$ & $\mathbf{0.94 \pm 0.74}$ \\
\hline
\end{tabular}
\end{adjustbox}
\end{center}
\vspace{-1.5em}
\label{tab:robustnessanalysis}
\end{table}
\normalsize

\subsection{Discussion}
From the results, it can be seen that, while the traditional inverse KKT estimate without any constraint adaption results in high errors, the estimates obtained by the truncated inverse KKT adaption as well as our proposed approach lead to an improved estimation performance (see Table \ref{tab:baseanalysis}, System 1). Furthermore, given the nature of the chosen constraint relaxation, which additionally allows for the consideration of wrongfully \emph{inactivated} inequality constraints, our proposed approach outperforms the truncated inverse KKT adaption for several noise realizations (see Table \ref{tab:baseanalysis} and System 2 and 3). Finally, the results in Table \ref{tab:robustnessanalysis} indicate empirical robustness to model uncertainty. 

\section{Conclusion}
This paper shows how to use penalty functions in inverse optimal control and demonstrates that correct estimation can be achieved for a problem description with and without penalization using optimal demonstrations. Furthermore, we demonstrate how leveraging this concept allows for a systematic reduction of unknown variables, as well as a relaxation of the primal feasibility and the complementary slackness condition in the inverse Karush-Kuhn-Tucker approach. Performance investigations for three different systems in simulation demonstrate that the presented relaxation approach enables high estimation accuracy, even in cases where the traditional inverse KKT approach is unsuccessful.  

\bibliographystyle{docstyle/IEEEtran} 
\bibliography{main}



\end{document}